\documentclass[letterpaper, 10 pt, conference]{ieeeconf}  

\IEEEoverridecommandlockouts                              

\usepackage{amsmath}
\usepackage{amssymb}

\usepackage{amsthm}
\usepackage{graphicx}
\usepackage{booktabs}
\let\labelindent\relax
\usepackage{enumitem}
\usepackage{pifont}
\usepackage{tikz}
\usepackage[tight,footnotesize]{subfigure}
\usepackage[ruled,vlined,linesnumbered]{algorithm2e}

\newtheorem{theorem}{Theorem}
\newtheorem{lemma}{Lemma}

\newtheorem{defn}{Definition}
\newtheorem{remark}{Remark}

\newtheorem*{problem*}{Problem}

\newcommand{\cmmnt}[1]{}

\title{\LARGE \bf
Leaderless collective motions in affine formation control
}

\author{Hector Garcia de Marina, Juan Jimenez Castellanos and Weijia Yao  
	\thanks{Hector Garcia de Marina and Juan Jimenez Castellanos are with the Faculty of Physics, Universidad Complutense de Madrid, 28040 Madrid, Spain. Weijia Yao is with ENTEG, University of Groningen, The Netherlands. {\tt\small hgarciad@ucm.es}, {\tt\small juan.jimenez@fis.ucm.es}, {\tt\small w.yao@rug.nl}.}%
}

\begin{document}
	\maketitle
	\thispagestyle{empty}
	\pagestyle{empty}

	\begin{abstract}
This paper proposes a novel distributed technique to induce collective motions in affine formation control. Instead of the traditional leader-follower strategy, we propose modifying the original weights that build the Laplacian matrix so that a designed steady-state motion of the desired shape emerges from the agents' local interactions. The proposed technique allows a rich collection of collective motions such as rotations around the centroid, translations, scalings, and shearings of a reference shape. These motions can be applied in useful collective behaviors such as \emph{shaped} consensus (the rendezvous with a particular shape), escorting one of the team agents, or area coverage. We prove the global stability and effectiveness of our proposed technique rigorously, and we provide some illustrative numerical simulations.
	\end{abstract}
	
	
	\section{Introduction}
The control of robot swarms is one of the grand challenges in robotics \cite{yang2018grand}. In particular, the display and maneuvering of geometrical patterns by mobile robots has been identified as one of the elementary blocks in swarm robotics \cite{Oh2015}. For example, the precise control of the geometrical pattern in the formation enables individuals to estimate the gradient of a field \cite{brinon2019multirobot}, to localize themselves relatively in \emph{GPS-denied} scenarios \cite{anderson2008rigid}, or to coordinate their relative motion \cite{de2019flexible}.

This paper shows how to coordinate the collective motion of a distributed formation while displaying the affine transformation of a reference shape. The \emph{affine formation control} algorithm guides the agents to converge to a static and arbitrary affine transformation of a reference shape \cite{lin2016necessary}. Such an algorithm is completely distributed where neighboring agents control \emph{tensions}, i.e., the weighted sum of all the sensed relative positions by an agent equals zero. This weighted sum of relative positions leads the weighted Laplacian matrix to emerge naturally during the analysis of the algorithm. Such weights can be designed so that the Laplacian is symmetric and positive semi-definite to assist the convergence to the \emph{affine shape} \cite{zhao2018affine}. However, the formation converges to a static position.

To maneuver the team of agents, classical strategies, such as the \emph{leader-follower}, have been proposed in \cite{zhao2018affine} although it requires an extra control layer superposed to the formation control algorithm. In this paper, we propose a different strategy to induce the collective motion of the formation. In particular, we show that the modification of weights, designed for the achievement of the \emph{static shape}, is sufficient to induce a collective motion. Such a collective behavior consists of linear combinations of \emph{affine motions} corresponding to affine transformations of the reference shape, such as rotations, shearings, translations, or scalings. A similar technique has been shown with 2D formations where agents encode their positions as complex numbers \cite{de2020distributed} or they exploit rotation matrices \cite{Marina2017}. In contrast to the complex formation control \cite{lin2014distributed}, the affine formation control enables us to jump from planar geometrical shapes ($\mathbb{C}^1$) to shapes in arbitrary dimensions ($\mathbb{R}^n$).

This paper has been divided into eight sections. Section \ref{sec: prel} explains the preliminaries including the adopted notation, employes concepts in graph theory, and the notion of \emph{reference} and \emph{desired} shape. Section \ref{sec: aff} reviews the algorithm for controlling affine formations. Sections \ref{sec: lap} and \ref{sec: man} present the maneuvering technique and how to design the motion parameters responsible for the collective behavior. Sections \ref{sec: sta} and \ref{sec: sim} address the stability analysis and illustrates the motions with numerical simulations. Finally, Section \ref{sec: dis} discusses future work and concludes the paper.

\section{Preliminaries}
\label{sec: prel}
\subsection{Notation and graph theory}
We consider $n \in \mathbb{N}$ mobile agents. We denote by $||x||$ the Euclidean norm of the vector $x\in\mathbb{R}^p, p\in\mathbb{N}$. Given a set $\mathcal{X}$, we denote by $|\mathcal{X}|$ its cardinality. We denote by $\mathbf{1}_p\in\mathbb{R}^p, p\in\mathbb{N}$, the all-one column vector. Finally, given a matrix $A\in\mathbb{R}^{p\times q},\, p,q\in\mathbb{N}$, we define the operator $\overline A := A \otimes I_m \in \mathbb{R}^{pm \times qm}$, where $\otimes$ denotes the Kronecker product, and $I_m\in\mathbb{R}^{m\times m}$ is the identity matrix.

A \emph{graph} $\mathcal{G} = (\mathcal{V}, \mathcal{E})$ consists of two non-empty sets: the node set $\mathcal{V} = \{1,2,\dots,n\}$, and the edge set $\mathcal{E} \subseteq (\mathcal{V}\times\mathcal{V})$. In this work, we only consider the special case of \emph{undirected} graphs. An undirected graphs is a \emph{bidirectional} graph where if the edge $(i,j)\in\mathcal{E}$, then the edge $(j,i)\in\mathcal{E}$ as well. The set $\mathcal{N}_i$ containing the neighbors of the node $i$ is defined by $\mathcal{N}_i:=\{j\in\mathcal{V}:(i,j)\in\mathcal{E}\}$. Let $w_{ij}\in\mathbb{R}$ be a weight associated with the edge $(i,j)\in\mathcal{E}$, then the \emph{Laplacian} matrix $L\in\mathbb{R}^{n\times n}$ of $\mathcal{G}$ is defined as
\begin{equation}
	l_{ij} := \begin{cases}\sum_{k\in\mathcal{N}_i}w_{ik} & \text{if} \quad i = j \\
		-w_{ij} & \text{if} \quad i \neq j \wedge j\in\mathcal{N}_i \\
		0 & \text{if} \quad i \neq j \wedge j\notin\mathcal{N}_i.
	\end{cases}
	\label{eq: L}
\end{equation}
We assume that $\mathcal{G}$ is connected so that $L\mathbf{1}_n = 0$. For an undirected graph, we choose one of the two arbitrary directions for each pair of neighboring nodes to construct the ordered set of edges $\mathcal{Z}$. For an arbitrary edge $\mathcal{Z}_k = (\mathcal{Z}_k^{\text{head}},\mathcal{Z}_k^{\text{tail}}), k\in\{1,\dots,\frac{|\mathcal{E}|}{2}\}$, we call its first and second element the \emph{head} and the \emph{tail} respectively. From such an ordered set, we construct the following \emph{incidence matrix} $B\in\mathbb{R}^{|\mathcal{V}|\times |\mathcal{Z}|}$ that satisfies $B^T\mathbf{1}_n = 0$:
\begin{equation}
	b_{ik} := \begin{cases}+1 \quad \text{if} \quad i = {\mathcal{Z}}^{\text{tail}}_k \\
		-1 \quad \text{if} \quad i = \mathcal{Z}^\text{head}_k \\
		0 \quad \text{otherwise.}
	\end{cases}
	\label{eq: B}
\end{equation}

\subsection{Affine desired shape}
Each agent $i\in\mathcal{V}$ has a position $p_i\in\mathbb{R}^m, \, m\in\mathbb{N}$. We stack all the positions $p_i$ in a single vector $p\in\mathbb{R}^{mn}$ and we call it \emph{configuration}. We define a \emph{framework} $\mathcal{F}$ as the pair $(\mathcal{G}, p)$, where we assign each agent's position $p_i$ to the node $i\in\mathcal{V}$, and the graph $\mathcal{G}$ establishes the set of neighbors $\mathcal{N}_i$ for each agent $i$.

We choose an arbitrary configuration of interest or \emph{reference shape $p^*$} for the team of agents, and we split it as
\begin{equation}
	p^* = \left(\mathbf{1}_n \otimes p_{\text{c.m.}}\right) + p^*_c,
	\label{eq: pstar}
\end{equation}
where $p_{\text{c.m.}}\in\mathbb{R}^{m}$ is the position of the \emph{center of mass} of the configuration and $p_c^*\in\mathbb{R}^{mn}$, starting from the center of mass, gives the \emph{appearance} to the formation as in the example shown in Figure \ref{fig: pstar}. Without loss of generality, and for the sake of simplicity, we set $p_{\text{c.m.}} = 0$ in (\ref{eq: pstar}). In this paper we assume that $p^*$ is \emph{generic} \cite{gortler2010characterizing}. For example, in 3D we do not set $p^*$ on the same plane.

\begin{figure}
\centering
\begin{tikzpicture}[line join=round]
\filldraw(0,0) circle (2pt);
\filldraw(1.5,0) circle (2pt);
\filldraw(1.5,1.5) circle (2pt);
\filldraw(0,1.5) circle (2pt);
\filldraw(2.25,0.75) circle (2pt);
\filldraw(-0.75,0.75) circle (2pt);
\filldraw(0.75,1.5) circle (2pt);
\filldraw(0.75,0) circle (2pt);
\draw[draw=black,arrows=->](-2,0)--(-1.5,0);
\draw[draw=black,arrows=->](-2,0)--(-2,0.5);
\draw[draw=black,arrows=->](-2,0)--(0.75,0.75);
\draw[draw=black,arrows=->](0.75,0.75)--(0+0.05,1.5-0.05);
\draw[draw=black,arrows=->](0.75,0.75)--(0+0.05,0+0.05);
\draw[draw=black,arrows=->](0.75,0.75)--(1.5-0.05,1.5-0.05);
\draw[draw=black,arrows=->](0.75,0.75)--(1.5-0.05,0+0.05);
\draw[draw=black,arrows=->](0.75,0.75)--(2.25-0.05,0.75);
\draw[draw=black,arrows=->](0.75,0.75)--(-0.75-0.05,0.75);
\draw[draw=black,arrows=->](0.75,0.75)--(0.75,1.5-0.05);
\draw[draw=black,arrows=->](0.75,0.75)--(0.75,+0.05);
\node at (-1.95,-0.25) {\small $O_g$ \normalsize};
\node at (0-0.2,1.2) {\small $p_{c_1}^*$ \normalsize}; 
\node at (0-0.2,0.75-0.5) {\small $p_{c_7}^*$ \normalsize}; 
\node at (1.5+0.2,0.75+0.45) {\small $p_{c_3}^*$ \normalsize}; 
\node at (1.5+0.2,0.75-0.4) {\small $p_{c_5}^*$ \normalsize}; 
\node at (2.25+0.2,0.75-0.2) {\small $p_{c_4}^*$ \normalsize}; 
\node at (-0.75-0.2,0.75+0.3) {\small $p_{c_8}^*$ \normalsize}; 
\node at (0.75+0.25,1.5-0.2) {\small $p_{c_2}^*$ \normalsize}; 
\node at (0.75+0.25,-0+0.3) {\small $p_{c_6}^*$ \normalsize}; 
\node at (-1.25,0.4) {\small $p_{\text{c.m.}}$ \normalsize}; 
\end{tikzpicture}
	\caption{Example of a reference 2D shape $p^* =  (\mathbf{1}_8 \otimes p_{\text{c.m.}}) + [p_{c_1}^{*T} \, \dots \, p_{c_8}^{*T}]^T$, where $p_{\text{c.m.}}$ is the center of mass of the reference shape.}
\label{fig: pstar}
\end{figure}
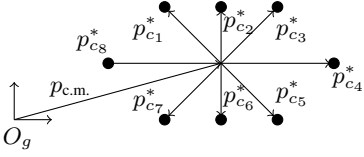

We now define the concept of \emph{desired shape} constructed from the reference shape $p^*$:
\begin{defn}
	The configuration is at the \emph{desired shape} when
	\begin{equation}
		p\in\mathcal{S}:=\{p = (I_n \otimes A)p^* + (\mathbf{1}_{n}\otimes\, b) \, | \, A\in\mathbb{R}^{m\times m}, b\in\mathbb{R}^m\}.
	\label{eq: dshape}
\end{equation}
\end{defn}
Note that the set $\mathcal{S}$ corresponds to all the affine transformations of the reference shape $p^*$; hence, the name of \emph{affine formation control} if our target is $p(t)\to\mathcal{S}$ as $t\to\infty$. The original work \cite{lin2016necessary} proposes an algorithm so that $p(t)$ converges to a point in $\mathcal{S}$. In this paper, we show how to modify slightly such an algorithm so that $p(t)$ converges to an \emph{orbit} or trajectory in $\mathcal{S}$.

\section{Agent dynamics and stabilization of an affine static shape}
\label{sec: aff}
The agents are modelled as point-mass particles where we can command their velocities, i.e.,
\begin{equation}
	\dot p_i = u_i, \quad i\in\mathcal{V},
	\label{eq: dyn}
\end{equation}
where $u_i\in\mathbb{R}^m$ is the control input to the agent $i$. Similarly, we can write the following compact form
\begin{equation}
\dot p = u,
\end{equation}
where $p,u\in\mathbb{R}^{nm}$ are the stacked vectors of positions and control actions respectively.

Since we want the nature of our algorithm to be distributed, then $u_i$ can only be a function of relative positions $z_{ij}:=(p_i-p_j),\, j\in\mathcal{N}_i$. In particular, the original algorithm that steers $p(t)\to\mathcal{S}$ has the following form \cite{lin2016necessary}
\begin{equation}
	u_i = -h \sum_{j\in\mathcal{N}_i} w_{ij} (p_i - p_j) = -h\sum_{j\in\mathcal{N}_i} w_{ij} z_{ij},
	\label{eq: uLaff}
\end{equation}
where $h\in\mathbb{R}_+$ is an arbitrary positive gain, and $w_{ij} = w_{ji}\in\mathbb{R}$ are weights, whose design together with the graph $\mathcal{G}$ will be explained shortly, so that the Laplacian matrix $L$ is positive semi-definite \cite{lin2016necessary,zhao2018affine}. Indeed, if we write the compact form of (\ref{eq: uLaff})
\begin{equation}
\dot p = -h\overline Lp,
	\label{eq: affcom}
\end{equation}
then $p(t) \to \operatorname{Ker}\{\overline L\}$. The kernel of $\overline L$ is the set $\mathcal{S}$ when we force the weights, besides the trivial solution, to satisfy
\begin{equation}
	\sum_{j\in\mathcal{N}_i} w_{ij} (p_{i}^* - p_{j}^*) = 0,\quad \forall i\in\mathcal{V},
	\label{eq: wdes}
\end{equation}
and such a set of weights can construct a positive semi-definite Laplacian matrix if and only if the framework is \emph{generically} and \emph{universally rigid} \footnote{Given a framework $(\mathcal{G}, p^*)$ with $p^*$ being generic, we say it is generically and universally rigid if for $(\mathcal{G}, q)$ with $q\in\mathbb{R}^s, \,s\in\mathbb{N}$, where $||p_i^* - p_j^*|| = ||q_i - q_j||, \forall (i,j) \in\mathcal{E}$, then it also implies that $||p_i^* - p_j^*|| = ||q_i - q_j||, \forall i,j \in\mathcal{V}$.} \cite{lin2016necessary}. This condition forces the number of agents to satisfy $n \geq m + 2$. 
We refer to \cite{kelly2014class} on how to build such frameworks in 2D and 3D, and to \cite{zhao2018affine} on how to compute the weights. If the framework is \emph{globally rigid} (relaxed condition) \cite{anderson2008rigid}, then $L$ can be made symmetric while the weights satisfy (\ref{eq: wdes}). However, in order to ensure that non-zero eigenvalues of $L$ are within the right-half complex plane, then we need to design a gain $k_i\in\mathbb{R}\setminus\{0\}$ for each agent $i$ that modifies (\ref{eq: uLaff})
as $u_i = -hk_i\sum_{j\in\mathcal{N}_i} w_{ij} z_{ij}$, or in compact form $\dot p = -h\overline K\overline Lp$, where $K := \operatorname{diag}\{k_1,\dots,k_n\}$. Note that such a matrix $K$ always exists so that $KL$ does not have eigenvalues with negative real part \cite{friedland1975inverse}.

\section{Modified Laplacian matrix}
\label{sec: lap}

In this section, we are going to show how to modify a (non-unique) subset of weights $w_{ij}$ in (\ref{eq: uLaff}) such that the configuration $p(t)$ converges as $t\to\infty$ to a steady-state motion within the desired shape $\mathcal{S}$.

Let us consider the following modified weights
\begin{equation}
	\tilde w_{ij} = w_{ij} - \frac{\kappa}{h}\mu_{ij}, \quad (i,j)\in\mathcal{E},
\label{eq: wmod}
\end{equation}
where the \emph{motion parameters} $\mu_{ij}\in\mathbb{R}$ will be designed shortly in Subsection \ref{sec: mus} for the translation, rotation, scaling, and shearing of the formation, and $\kappa\in\mathbb{R}$ will regulate the speed of the collective motion. The gain $h$ is included in (\ref{eq: wmod}) to compensate itself once the new modified Laplacian is multiplied by $h$ as in (\ref{eq: affcom}). Since our maneuvering technique is distributed, then $\mu_{ij} = 0$ if $j\notin\mathcal{N}_i$, and in general $\mu_{ij} \neq \mu_{ji}$.

Similarly to the incidence matrix $B$ in (\ref{eq: B}), consider again the ordered set of edges $\mathcal{Z}$, and let us define the components of the following matrix $M\in\mathbb{R}^{|\mathcal{V}|\times |\mathcal{Z}|}$
\begin{equation}
	m_{ik} := \begin{cases}\mu_{i\mathcal{Z}^\text{head}_k} \quad \text{if} \quad i = \mathcal{Z}^\text{tail}_k \\
		-\mu_{i\mathcal{Z}^\text{tail}_k} \quad \text{if} \quad i = \mathcal{Z}^\text{head}_k \\
		0 \quad \text{otherwise.}
	\end{cases}.
	\label{eq: M}
\end{equation}

The definition (\ref{eq: M}) enables us to write the modified Laplacian matrix from the modified weights (\ref{eq: wmod}) in compact form as
\begin{equation}
	\tilde L = L - \frac{\kappa}{h}MB^T.
	\label{eq: Ltilde}
\end{equation}

\section{Affine maneuvering}
\label{sec: man}
\subsection{Motion parameters design}
\label{sec: mus}
The motion parameters $\mu_{ij}$ are designed in a similar way as the weights $w_{ij}$ in (\ref{eq: wdes}), i.e., they must satisfy the following linear contraints
\begin{equation}
	{^b}v_i^* = \sum_{j\in\mathcal{N}_i}\mu_{ij}\, ({^b}p_i^* - {^b}p_j^*)  = \sum_{j\in\mathcal{N}_i}\mu_{ij}\,{^b}z_{ij}^*, \, \forall i\in\mathcal{V},
\label{eq: vbi}
\end{equation}
where ${^b}v_i^*\in\mathbb{R}^m$ is the desired velocity for each agent $i$ so that the collective motion is compatible with having $p(t)\in\mathcal{S}$ if we start from $p^*$ as it is shown in Figure \ref{fig: sqdis}. Note that the desired agent's velocity is designed with respect to a frame of coordinates $O_b$ at the center of mass of $p^*$; hence, the $^b$ superscript. Since the desired agents' velocities are constructed from the relative positions $z^* = \overline B^Tp^*$, we will see that if all the stacked vectors in $z^*$ go under an affine transformation, then the resultant motion will be transformed equally as well. For example, we can design a circular motion (rotation) around the centroid of $p^*$. However, after applying an affine transformation to the shape described by $p^*$, then the resultant rotation will be ellipsoidal in general.

\begin{figure}
\centering
\begin{tikzpicture}[line join=round]
\begin{scope}[shift={(-4,0)}]
\draw[dashed](1.5,0)--(0,0)--(0,1.5)--(1.5,1.5)--(1.5,0)--(0,1.5);
\draw[dashed](0,0)--(1.5,1.5);
\filldraw(0,0) circle (2pt);
\filldraw(1.5,0) circle (2pt);
\filldraw(1.5,1.5) circle (2pt);
\filldraw(0,1.5) circle (2pt);
\draw[draw=black,arrows=->](.75,.75)--(1.125,.75);
\draw[draw=black,arrows=->](.75,.75)--(.75,1.125);
\draw[draw=black,arrows=->](-.75,1.0)--(-1.125,1.0);
\draw[draw=black,arrows=->](-.75,1.0)--(-.75,0.6);
	\node at (.8,.45) {\small $O_b$ \normalsize};\node at (-.75,1.2) {\small $O_g$ \normalsize};\node at (-.25,-.25) {\small $p_1^*$ \normalsize}; \node at (1.25,1.75) {\small $p_3^*$ \normalsize};\node at (-.25,1.75) {\small $p_2^*$ \normalsize}; \node at (1.25,-.25) {\small $p_4^*$ \normalsize};
\end{scope}
\begin{scope}[shift={(-1.25,0)},scale=0.35]
\filldraw(0,0) circle (2pt);
\filldraw(1.5,0) circle (2pt);
\filldraw(1.5,1.5) circle (2pt);
\filldraw(0,1.5) circle (2pt);
\draw[dashed](1.5,0)--(0,0)--(0,1.5)--(1.5,1.5)--(1.5,0)--(0,1.5);
\draw[dashed](0,0)--(1.5,1.5);
\draw[draw=black,color=blue,arrows=->](0,0)--(1,0);
\draw[draw=black,color=blue,arrows=->](1.5,0)--(2.5,0);
\draw[draw=black,color=blue,arrows=->](1.5,1.5)--(2.5,1.5);
\draw[draw=black,color=blue,arrows=->](0,1.5)--(1,1.5);
\end{scope}
\begin{scope}[shift={(-1.25,1.25)},scale=0.35]
\filldraw(0,0) circle (2pt);
\filldraw(1.5,0) circle (2pt);
\filldraw(1.5,1.5) circle (2pt);
\filldraw(0,1.5) circle (2pt);
\draw[dashed](1.5,0)--(0,0)--(0,1.5)--(1.5,1.5)--(1.5,0)--(0,1.5);
\draw[dashed](0,0)--(1.5,1.5);
\draw[draw=black,color=blue,arrows=->](0,0)--(0,1);
\draw[draw=black,color=blue,arrows=->](1.5,0)--(1.5,1);
\draw[draw=black,color=blue,arrows=->](1.5,1.5)--(1.5,2.5);
\draw[draw=black,color=blue,arrows=->](0,1.5)--(0,2.5);
\end{scope}
\begin{scope}[shift={(0,0)},scale=0.35]
\filldraw(0,0) circle (2pt);
\filldraw(1.5,0) circle (2pt);
\filldraw(1.5,1.5) circle (2pt);
\filldraw(0,1.5) circle (2pt);
\draw[dashed](1.5,0)--(0,0)--(0,1.5)--(1.5,1.5)--(1.5,0)--(0,1.5);
\draw[dashed](0,0)--(1.5,1.5);
\draw[draw=black,color=red,arrows=->](0,0)--(1,-1);
\draw[draw=black,color=red,arrows=->](1.5,0)--(2.5,1);
\draw[draw=black,color=red,arrows=->](1.5,1.5)--(0.5,2.5);
\draw[draw=black,color=red,arrows=->](0,1.5)--(-1,0.5);
\end{scope}
\begin{scope}[shift={(0,1.25)},scale=0.35]
\filldraw(0,0) circle (2pt);
\filldraw(1.5,0) circle (2pt);
\filldraw(1.5,1.5) circle (2pt);
\filldraw(0,1.5) circle (2pt);
\draw[dashed](1.5,0)--(0,0)--(0,1.5)--(1.5,1.5)--(1.5,0)--(0,1.5);
\draw[dashed](0,0)--(1.5,1.5);
\draw[draw=black,color=green,arrows=->](0,0)--(-1,-1);
\draw[draw=black,color=green,arrows=->](1.5,0)--(2.5,-1);
\draw[draw=black,color=green,arrows=->](1.5,1.5)--(2.5,2.5);
\draw[draw=black,color=green,arrows=->](0,1.5)--(-1,2.5);
\end{scope}
\begin{scope}[shift={(1.25,0)},scale=0.35]
\filldraw(0,0) circle (2pt);
\filldraw(1.5,0) circle (2pt);
\filldraw(1.5,1.5) circle (2pt);
\filldraw(0,1.5) circle (2pt);
\draw[dashed](1.5,0)--(0,0)--(0,1.5)--(1.5,1.5)--(1.5,0)--(0,1.5);
\draw[dashed](0,0)--(1.5,1.5);
\draw[draw=black,color=orange,arrows=->](0,0)--(1,0);
\draw[draw=black,color=orange,arrows=->](1.5,0)--(2.5,0);
\draw[draw=black,color=orange,arrows=->](1.5,1.5)--(0.5,1.5);
\draw[draw=black,color=orange,arrows=->](0,1.5)--(-1,1.5);
\end{scope}
\begin{scope}[shift={(1.25,1.25)},scale=0.35]
\filldraw(0,0) circle (2pt);
\filldraw(1.5,0) circle (2pt);
\filldraw(1.5,1.5) circle (2pt);
\filldraw(0,1.5) circle (2pt);
\draw[dashed](1.5,0)--(0,0)--(0,1.5)--(1.5,1.5)--(1.5,0)--(0,1.5);
\draw[dashed](0,0)--(1.5,1.5);
\draw[draw=black,color=orange,arrows=->](0,0)--(0,1);
\draw[draw=black,color=orange,arrows=->](1.5,0)--(1.5,-1);
\draw[draw=black,color=orange,arrows=->](1.5,1.5)--(1.5,0.5);
\draw[draw=black,color=orange,arrows=->](0,1.5)--(0,2.5);
\end{scope}
\end{tikzpicture}
	\caption{Square formation in 2D with $\mathcal{E}$ derived from a complete graph with four nodes so that the framework is universally rigid. The reference shape $p^*$ for the affine formation control is designed with respect to a global frame of coordinates $O_g$. However, the set of affine collective motions is designed with respect to a frame of coordinates $O_b$ attached at the centroid. In fact, the desired velocity vector $^bv_i^*$ for the agent $i$ is constructed as a linear combination of the relative positions $z_{ij}, j\in\mathcal{N}_i$ employing the motion parameters $\mu_{ij}$. On the right side, we choose a non-orthogonal basis to construct the affine collective motions in 2D, e.g., two orthogonal translation velocities in blue color, spinning around the centroid in red color, scaling in green color, and two orthogonal shearing motions in orange color.}
	\label{fig: sqdis}
\end{figure}
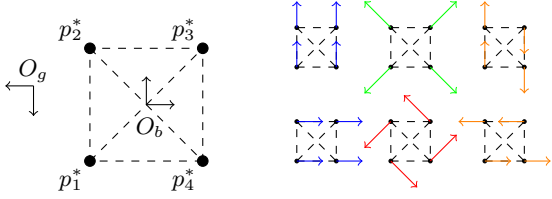

Note that in order to find the $\mu_{ij}$'s that satisfy (\ref{eq: vbi}) for an arbitrary ${^b}v_i^*$, it is sufficient for the agent $i$ to have at least $m$ neighbors with the corresponding $z_{ij}^*$ being linearly independent. Indeed, this is the case if $p^*$ is \emph{generic} and the framework globally rigid.

For example, let us show the design of $\mu_{ij}$ for the rotational motion in Figure \ref{fig: sqdis}. It is clear that $^bv_1^* = [\begin{smallmatrix}1 & -1\end{smallmatrix}]^T, {^bv}_2^* = [\begin{smallmatrix}-1 & -1\end{smallmatrix}]^T, {^bv}_3^* = [\begin{smallmatrix}-1 & 1\end{smallmatrix}]^T$ and $^bv_4^* = [\begin{smallmatrix}1 & 1\end{smallmatrix}]^T$ up to an arbritary scale (angular speed) factor. In order to satisfy (\ref{eq: vbi}), if we consider the square's side equals one, then the motion parameters of the agent $1$ for the rotational motion are $\mu_{12} = -\mu_{14} = 1$, and $\mu_{13} = 0$, since $(p^*_1 - p^*_2) = [\begin{smallmatrix}0 & -1\end{smallmatrix}]^T$, and $(p^*_1 - p^*_4) = [\begin{smallmatrix}-1 & 0\end{smallmatrix}]^T$. Note that this is not the only choice since we have not used $(p^*_1 - p^*_3) = [\begin{smallmatrix}-1 & -1\end{smallmatrix}]^T$.

We can stack (\ref{eq: vbi}) for all the agents and arrive at the following compact form
\begin{equation}
	{^b}v_f^* = \overline M\overline B^T\,{^b}p^*,
	\label{eq: vf}
\end{equation}
where ${^b}v_f^* \in \mathbb{R}^{mn}$ is the stacked vector with all the desired agents' velocities.

\subsection{Collective behaviors in 2D and 3D}
\label{sec: beh}
To assist the design for the eventual collective behavior of the formation, we can split ${^b}v_f^*$ into scaling, rotations around the centroid, translations, and shearings. In a extension of this work, we will see rigorously that such velocities form a basis for all the possible motions that keep a configuration in $\mathcal{S}$ if $p(0)\in\mathcal{S}$. In 2D, we have six velocities forming a (non-orthogonal) basis as shown in Figure \ref{fig: sqdis}, while in 3D, we have three translations, three rotations, three shearings, and one scaling collective velocities. In particular, for the 2D case, we can split $M$ in (\ref{eq: vf}) into
\begin{equation}
	M = \kappa_{t_1}M_{t_1} + \kappa_{t_2}M_{t_2} +\kappa_s M_s + \kappa_rM_r + \kappa_{s_1}M_{s_1} + \kappa_{s_2}M_{s_2},
	\label{eq: ormo}
\end{equation}
where the matrices $M_{t_{\{1,2\}}},M_r,M_s,M_{s_{\{1,2\}}}\in\mathbb{R}^{|\mathcal{V}|\times|\mathcal{Z}|}$ have their elements $\mu_{ij}$ as in (\ref{eq: M}) but designed only for the $1$ \emph{distance units}/sec translation in horizontal/vertical direction, for the $1$ \emph{current size}/sec scaling, for the $1$ radian/sec spinning, and for the $1$ \emph{distance units}/sec shearing in horizontal/vertical direction of the reference shape respectively. Finally, we can see $\kappa_{\{t_{\{1,2\}},r,s,s_{\{1,2\}}\}}\in\mathbb{R}$ as the \emph{coordinates} of the six motions (see the right side in Figure \ref{fig: sqdis}) that will define the eventual collective motion. For example, if $\kappa_{\{r,s\}} = -1$ and $\kappa_{t_{\{1,2\}}}=\kappa_{s_{\{1,2\}}}=0$, then, the eventual collective motion will be the contraction of an affine transformation of the reference shape while its centroid is fixed but the agents will orbit around it.

Looking at (\ref{eq: vf}), we have that the designed $\mu_{ij}$ for the collective motions that form a basis satisfy the following expressions, focusing both in 2D and 3D. Firstly, regarding the pure translations we have that
\begin{equation}
	\sum_{l=1}^{m}\overline M_{t_l}\overline B^Tp^* = (\mathbf{1}_n \otimes v^*),
	\label{eq: mutran}
\end{equation}
where $v^*\in\mathbb{R}^m$ is the common desired translational velocity for all the agents. Secondly, for the rotations around the centroid of $p^*$ we can check that
\begin{equation}
	\overline M_r\overline B^Tp^* = (I_n \otimes W)p^*,
	\label{eq: murot}
\end{equation}
where $W\in\mathbb{R}^{m\times m}$ is the \emph{angular velocity tensor}, e.g., for the 3D case $W = \left[\begin{smallmatrix}0 & -\omega_z & \omega_y \\ \omega_z & 0 & -\omega_x \\ -\omega_y & \omega_x & 0\end{smallmatrix}\right]$, with $\omega_{\{x,y,z\}}=1$ rads/sec being the angular velocities around the three Cartesian axes. Note that $M_r$ in (\ref{eq: murot}) is the superposition of all the orthogonal angular velocity tensors forming a basis, i.e., for 2D we can only rotate around the \emph{vertical} axis but in 3D we have three possible rotations. Thirdly, for the scaling of $p^*$, the design of $M_s$ satisfies
\begin{equation}
	\overline M_s\overline B^Tp^* = p^*,
	\label{eq: musca}
\end{equation}
where the \emph{scaling velocity tensor} is the identity matrix.
Finally, for the shearing motion of $p^*$, the following is satisfied
\begin{equation}
	\sum_{l=1}^{m}\overline M_{s_l}\overline B^Tp^* = (I_n \otimes S)p^*,
	\label{eq: mushe}
\end{equation}
where $S\in\mathbb{R}^{m\times m}$ is the \emph{shearing velocity tensor}. For example, for the 2D case $S = \left[\begin{smallmatrix}0 & h_{xy} \\ h_{yx} & 0\end{smallmatrix}\right]$ with $h_{xy}=h_{yx}=1$ distance units/sec defining the shearing velocity parallel to the $x$-axis and $y$-axis (in $O_b$) respectively.
\begin{remark}
As an example of non-orthogonality between the proposed motions, one can achieve a 2D rotational motion by combining two shearing velocities, e.g., by setting $\kappa_{s_1} = -\kappa_{s_2} = 1$. Nevertheless, the chosen collective motions,  albeit non-orthogonal, form a basis for all the affine motions and they have a straightforward physical meaning. A more detailed analysis on an orthogonal basis for the velocities will be covered in the journal extension of this work.
\end{remark}


Before the main result, we need one technical lemma. The following statement proves that if we apply the proposed motion technique to an affine transformation of the shape described by $p^*$, the result is another affine transformation of the shape described by $p^*$. It is particularly an affine transformation of the stacked designed velocities in $v^*_f$.
\begin{lemma}
	\label{lem: Av}
	Consider the affine transformation $\hat p = (I_n\otimes A)p^* + (\mathbf{1}\otimes b)$, for arbitrary $A\in\mathbb{R}^{m\times m}$ and $b\in\mathbb{R}^m$, then $\overline M\overline B^T \hat p = (I_n\otimes A)v_f^* \in\mathcal{S}.$
\end{lemma}
\begin{proof}
	The following identities exploit the \emph{mixed-product} property $(X_1 \otimes X_2)(X_3 \otimes X_4) = (X_1X_3 \otimes X_2X_4)$ for four matrices $X_{\{1,2,3,4\}}$ whose dimensions allow the latter matrix multiplications. In particular $(X_1 \otimes I)(I \otimes X_4) = (X_1 \otimes X_4) = (I \otimes X_4)(X_1 \otimes I)$ with the appropriate dimensions for the identity matrix.

Consider $M_t$ as the superposition (or sum) of all $M_{t_l}$ accounting for all the translational motions in (\ref{eq: mutran}), then we have that
\begin{align}
	\overline{M_t}\overline B^T\hat p &= (M_t \otimes I_m)(B^T \otimes I_m)(I_n \otimes A)p^* \nonumber \\
	&= (M_t \otimes I_m)(I_n \otimes A)(B^T \otimes I_m)p^* \nonumber \\
	&= (I_n \otimes A)(M_t \otimes I_m)(B^T \otimes I_m)p^* \nonumber \\
	&= (I_n \otimes A)(\mathbf{1}_n \otimes v^*) \nonumber \\
	&= \mathbf{1}_n \otimes Av^*,
\end{align}
	where, of course, $v^*$ is as in (\ref{eq: mutran}). We can conclude then that $\overline{M_t}\overline B^T\hat p\in\mathcal{S}$. Similarly, for all the rotational motions in (\ref{eq: murot}), we have that
\begin{align}
	\overline{M_r}\overline B^T\hat p &= (M_r \otimes I_m)(B^T \otimes I_m)(I_n \otimes A)p^* \nonumber \\
	&= (I_n \otimes A)(M_r \otimes I_m)(B^T \otimes I_m)p^* \nonumber \\
	&= (I_n \otimes A)(I_n \otimes W)p^* \nonumber \\
	&= (I_n \otimes AW)p^*,
\end{align}
	consequently, we have that $\overline{M_r}\overline B^T\hat p \in\mathcal{S}$. Accordingly, we also have for the scaling and shearing motions from (\ref{eq: musca}) and (\ref{eq: mushe}) that $\overline{M_s}\overline B^T\hat p = (I_n \otimes A)p^*\in\mathcal{S}$, and $\sum_{l=1}^{m}\overline{M_{s_l}}\overline B^T\hat p = (I_n \otimes AS)p^*\in\mathcal{S}$.
\end{proof}

\section{Stability analysis}
\label{sec: sta}
In this section, we replace the original weights in the affine formation controller (\ref{eq: uLaff}), that achieve a static shape in $\mathcal{S}$, with the modified weights in (\ref{eq: wmod}) designed from the motion parameters $\mu_{ij}$ derived in Section \ref{sec: mus}. We then present the following controller
\begin{equation}
	u_i = - h\sum_{j\in\mathcal{N}_i} \tilde w_{ij} (p_i - p_j).
	\label{eq: udismu}
\end{equation}
Similarly as in (\ref{eq: affcom}), using (\ref{eq: Ltilde}), we can arrive at the following compact form.
\begin{align}
	\dot p &= -h \overline{\tilde L}p = -h\overline Lp - \kappa \overline M \overline B^T p.
\label{eq: claffm}
\end{align}

Differently from \cite{de2020distributed} for the $\mathbb{C}^1$ case, here for the $\mathbb{R}^m$ case, we are not going to focus on the explicit solutions of the closed loop (\ref{eq: claffm}). We leave, for the extension of this work, the detailed analysis for the modification of the original $m^2+m$ zero eigenvalues of $\overline L$ and their corresponding eigenvectors. Nonetheless, we are going to show via Lypaunov that $p(t) \to \mathcal{S}$ and $\dot p(t) \to \kappa \overline M\overline B^T p(t)$ simultaneously as $t\to\infty$.

\begin{theorem}
	\label{th: main}
	Consider a generic reference configuration $p^*\in\mathbb{R}^{mn}$ and a framework $\mathcal{F}$ universally rigid such that we can find weights $w_{ij} = w_{ji}$ for the (positive semi-definite) $L$ as in (\ref{eq: L}), so that $\operatorname{Ker}\{\overline L\}$ is as in (\ref{eq: dshape}). Consider the control action (\ref{eq: udismu}) for the dynamics (\ref{eq: dyn}) where the modified weights are as in (\ref{eq: wmod}). If
	$$
h > \kappa \Big|\Big|Q\left(\overline{M}\overline B^T\right)^\dagger \Big|\Big|_2,
	$$
where the symbol $\dagger$ means that we take the last $(m(n-m-1) \times m(n-m-1))$ diagonal block of the square matrix, $Q$ is a positive definite matrix such that $QJ_2 + J_2^TQ = 2I$, and $J_2\in\mathbb{R}^{m(n-m-1)\times m(n-m-1)}$ is the Jordan form for the non-zero eigenvalues of $\overline L$, then $p(t) \to \mathcal{S}$ and $\dot p(t) \to \kappa \overline M\overline B^T p(t)$ as $t\to\infty$.
\end{theorem}
\begin{proof}
	We define $\mathcal{S}^\perp$ as the orthogonal space of $\mathcal{S}$, where both subspaces have dimensions $m(n-m-1)$ and $m^2+m$ respectively \cite{lin2016necessary}, and we remind that $\mathcal{S} = \operatorname{Ker}\{\overline L\}$. Let us split
\begin{equation}
	p = P_\mathcal{S}\,p + P_\mathcal{S^\perp}\,p = p_\parallel + p_\perp,
\end{equation}
where $P_\mathcal{X}$ stands for the projection matrix over the space $\mathcal{X}$. We are going to show the convergence of $p_\perp (t) \to 0$ as $t\to\infty$. Therefore, we will have that $p(t)$ has $p_\parallel$ as the only nonzero components eventually, i.e., $p(t)\to\mathcal{S}$, as $t\to\infty$. We write the dynamics for $p_\perp$ derived from the closed loop (\ref{eq: claffm}) as
\begin{equation}
	\dot p_\perp = -hP_\mathcal{S^\perp}\overline L(p_\parallel + p_\perp) + \kappa P_\mathcal{S^\perp} \overline M\overline B^T(p_\parallel + p_\perp).
\label{eq: aux98}
\end{equation}
According to Lemma \ref{lem: Av}, we have that
	\begin{equation}
	\kappa(\overline{M}\overline B^T)\big((\mathbf{1}_n \otimes A)p^* + \mathbf{1}_n \otimes b\big) \in \mathcal{S},
	\end{equation}
	i.e., $\kappa(\overline{M}\overline B^T)p_\parallel \in \mathcal{S}$. Together with $\overline L p_\perp \in \mathcal{S}^\perp$, and $\overline L p_\parallel = 0$, we can further simplify (\ref{eq: aux98}) as
\begin{equation}
\dot p_\perp = -h\overline L p_\perp + \kappa P_\mathcal{S^\perp} \overline{M}\overline B^T p_\perp.
	\label{eq: aux103}
\end{equation}
	Now consider the Jordan form $J\in\mathbb{R}^{mn\times mn}$ of $\overline L$, i.e., $T\overline L T^{-1} = \left[\begin{smallmatrix}J_1 & 0 \\ 0 & J_2\end{smallmatrix}\right]$ for some invertible matrix $T$, and $J_1\in\mathbb{R}^{(m^2+m)\times (m^2+m)}$ and $J_2\in\mathbb{R}^{m(n-m-1) \times m(n-m-1)}$, where we consider $J_1$ the zero matrix corresponding to the zero eigenvalues of $\overline L$. Let us apply the change of coordinates $T$ to $p$, i.e.,
\begin{equation}
\begin{bmatrix}q_1 \\ q_2\end{bmatrix} = Tp,
\end{equation}
	with $q_1\in\mathbb{R}^{m^2+m}$ and $q_2\in\mathbb{R}^{m(n-m-1)}$. Note that $Tp_\parallel = \begin{bmatrix}q_1^T & 0\end{bmatrix}^T$ and  $Tp_\perp = \begin{bmatrix}0 & q_2^T\end{bmatrix}^T$. Then, by applying the same coordinate transformation to (\ref{eq: aux103}), we have that
\begin{align}
	\frac{\mathrm{d}}{\mathrm{dt}}\begin{bmatrix}0 \\ q_2\end{bmatrix} &= -hT\overline LT^{-1}\begin{bmatrix}0 \\ q_2\end{bmatrix} + \kappa TP_{\mathcal{S}^\perp}\overline{M}\overline B^TT^{-1}\begin{bmatrix}0 \\ q_2\end{bmatrix} \nonumber \\
		\dot q_2 &= -hJ_2 q_2 + \kappa \left(TP_{\mathcal{S}^\perp}\overline{M}\overline B^TT^{-1}\right)^\dagger q_2,
	\label{eq: aux105}
\end{align}
	where the symbol $\dagger$ means that we take the last $(m(n-m-1) \times m(n-m-1))$ diagonal block of the matrix to accommodate for the dimensions of $q_2$. If we show that $q_2(t)\to 0$ as $t\to\infty$, then $p_\perp(t)\to 0$ as $t\to\infty$. Let us choose the following Lyapunov function $V = q_2^TQq_2$, where $Q$ is a positive definite matrix such that $QJ_2 + J_2^TQ = 2I_{m(n-m-1)}$. Then, the time derivative of $V$ satisfies
\begin{align}
	\frac{\mathrm{d}V}{\mathrm{dt}} \leq -2h||q_2||^2 +2\kappa \, ||Q\left(TP_{\mathcal{S}^\perp} \overline M\overline B^TT^{-1}\right)^\dagger ||_2 \, ||q_2||^2. \nonumber
\end{align}
If we exploit the fact that to do and undo the change of coordinates with $T$ does not change the norm of a vector, and that the projection matrix does not make bigger the norm of a vector, then if we choose $h$ such that
\begin{equation}
	h > \kappa \Big|\Big|Q\left(\overline{M}\overline B^T\right)^\dagger \Big|\Big|_2, \label{eq: kapaf}
\end{equation}
we have that $p_\perp(t) \to 0$ as $t \to \infty$ exponentially fast, i.e., $p(t)\to p_{\parallel}(t)\in\mathcal{S}$ as $t\to\infty$. Therefore, we conclude that if $h$ satisfies (\ref{eq: kapaf}), then, we can deduce from the closed loop (\ref{eq: claffm}) that
\begin{equation}
	\begin{cases}
		p(t) \to & p_{\parallel}(t) \\
		\dot p(t) \to & \kappa \overline{M}\overline B^Tp(t) = \kappa \overline{M}\overline B^Tp_{\parallel}(t)
	\end{cases}, \, t\to\infty
	\label{eq: velaff}
\end{equation}
\end{proof}

Theorem \ref{th: main} has shown us that the dynamics of $p(t)$ converges to the linear system (\ref{eq: velaff}) together with $p(t)\to\mathcal{S}$. Consequently, we can deduce the eventual collective behavior by analyzing the linear system (\ref{eq: velaff}) whose initial condition is a configuration in $\mathcal{S}$. Note that the global convergence to the desired collective behavior is guaranteed.

Physically, we can understand the results of Theorem \ref{th: main} better if we split $p(t)$ into two configurations, namely, $p_c(t)$ and $p_{\text{c.m.}}(t)$. The configuration $p_c(t)$ is as in (\ref{eq: pstar}) (we remind that we defined as reference shape $p^* = p_c^*$) with collective motions that keep its centroid fixed, and the \emph{centroid configuration} $p_{\text{c.m.}}(t)$ travels depending on the actual $p_c(t)$. Hence, the eventual collective behavior can be formally expressed as
\begin{equation}
\begin{cases}
	p(t) &= p_{\text{c.m.}}(t) + p_c(t) \\
	\dot p_c(t) &= \kappa\big(I_n \otimes (W + I_m + H)\big)\overline B^Tp_c(t) \\
	\dot p_{\text{c.m.}}(t) &= \kappa\overline M_t\overline B^T p_c(t) \\
	p_c(0)&\in\{(I_n\otimes A)p^* \, | \, A\in\mathbb{R}^{m\times m}\}\\
	p_{\text{c.m.}}(0) &\in \{\mathbf{1}\otimes b \, | \, b\in\mathbb{R}^m \}
\end{cases},
	\label{eq: linsys}
\end{equation}
where the particular $A$ and $b$ to pick in (\ref{eq: linsys}) will depend on the initial condition $p(0)$ in (\ref{eq: claffm}), and for the sake of simplicity we have assumed that all the coordinates $\kappa_*$ (see Subsection \ref{sec: beh}) equal one.

\section{Numerical simulations}
\label{sec: sim}
We choose as a reference shape $p^*$ the one displayed in Figure \ref{fig: pstar}. In particular, the separations between agents in the horizontal and vertical axes are equal to $1$. We create an universally rigid framework by setting the collection of edges as
\begin{align}
	\mathcal{Z} &= \{(1,2),(1,3),(1,4),(1,5),(2,4),(2,7),(3,5) \nonumber \\ & (3,6),(4,5),(4,6),(5,7),(6,8),(7,8),(4,8),(5,8)\}. \nonumber
\end{align}
where the weights $w_{ij}=w_{ji}$ have been calculated following the algorithm in \cite{zhao2018affine}. We describe in detail four simulations with collective behaviors based on $p^*$ in the captions of Figures \ref{fig: aff1}, \ref{fig: aff2}, \ref{fig: aff3}, and \ref{fig: aff4}.

\section{Discussion and future work}
\label{sec: dis}
We have presented how to induce collective motions in affine formation control. These collective behaviors do not require leaders but to modify the original weights responsible for only a static configuration. As illustrated in the simulations, we can exploit these behaviors to rendezvous the agents with a particular shape, enclose a point of interest, or cover an area. Similarly as in \cite{de2020distributed}, future work will focus on obtaining the explicit solutions to the closed-loop system by analyzing eigenvalues and eigenvectors of $\overline {\tilde L}$ in (\ref{eq: claffm}).

\begin{figure}
\centering
\includegraphics[width=1.0\columnwidth]{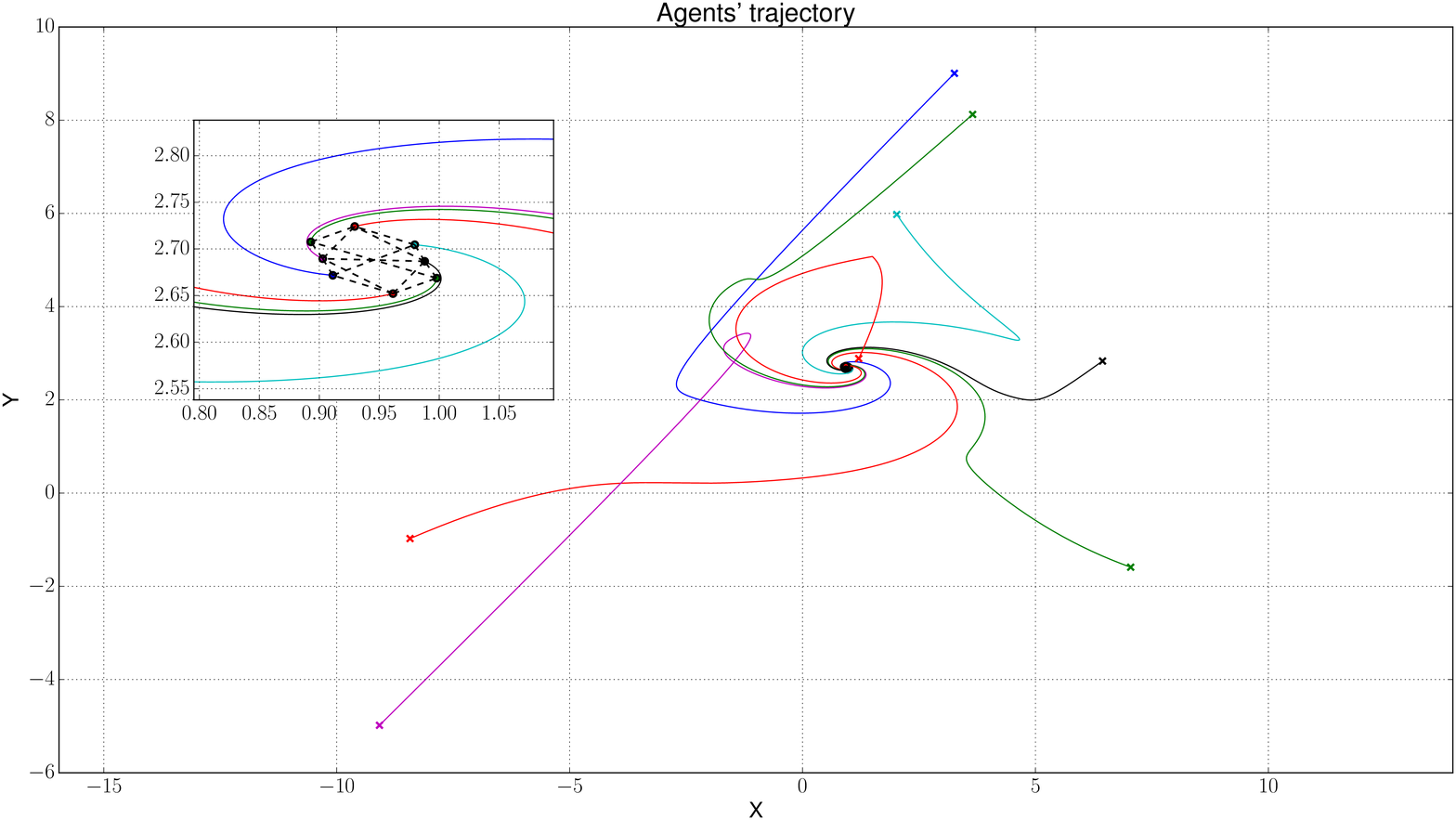}
	\caption{We take the reference shape $p^*$ as in Figure \ref{fig: pstar}. The crosses and the dots denote for the initial and the $t = 250$ secs positions. The dashed lines denote the edges of the graph. In (\ref{eq: ormo}), we set all the coordinates to zero except that $\kappa_s = \kappa_r = -1$. This collective behavior can be regarded as the \emph{shaped consensus} where the formation rendezvous while describing an affine version of $p^*$. Note that the rotation is not circular but ellipsoidal, and together with the negative scaling the agents describe an inwards spiral.}
\label{fig: aff1}
\end{figure}

\begin{figure}
\centering
\includegraphics[width=1\columnwidth]{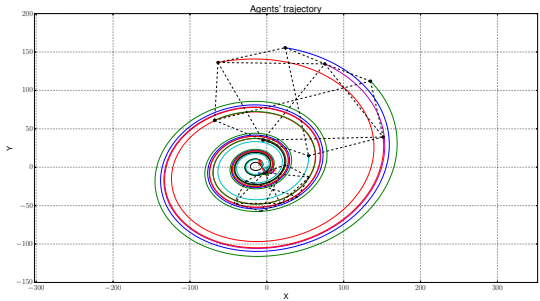}
	\caption{We take the reference shape $p^*$ as in Figure \ref{fig: pstar}. The crosses and the dots denote for the initial and the $t = 250$ secs positions. The dashed lines denote the edges of the graph. We set $\kappa_s = \kappa_r = \kappa_{t_1} = 1$, while the rest coordinates in (\ref{eq: ormo}) are zero. Recall that the scaling velocity tensor defines an exponentially growing speed. Since the translational velocity of the formation depends on the size, it also grows exponentially fast. We captured this fact by showing another snapshot of the formation with lower scale at $t = 175$ secs. Note that between $t=175$ and $t=250$ seconds the formation travelled more distance than between $t=0$ and $t=175$ secs.}
\label{fig: aff2}
\end{figure}

\begin{figure}
\centering
\includegraphics[width=1.0\columnwidth]{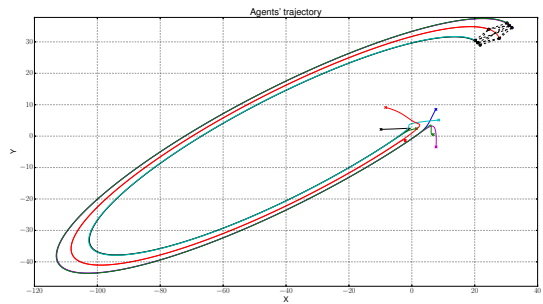}
	\caption{We take the reference shape $p^*$ as in Figure \ref{fig: pstar}. The crosses and the dots denote for the initial and the $t = 350$ secs positions. The dashed lines denote the edges of the graph. We set the coordinates $\kappa_r = \kappa_{t_1} = 1$, while the rest are zero in (\ref{eq: ormo}). The formation converges to a closed ellipsoidal orbit which is an affine transformation of the designed circular trajectory.}
\label{fig: aff3}
\end{figure}

\begin{figure}
\centering
\includegraphics[width=1\columnwidth]{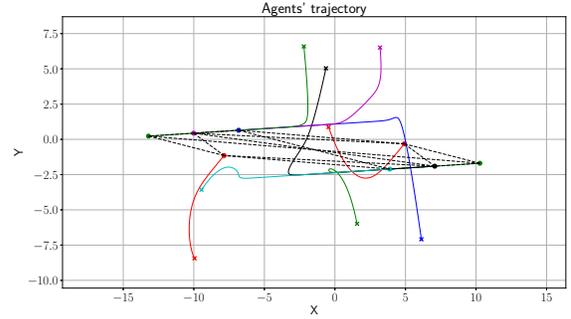}
	\caption{We take the reference shape $p^*$ as in Figure \ref{fig: pstar}. The crosses and the dots denote for the initial and the $t = 250$ secs positions. The dashed lines denote the edges of the graph. We design a positive shearing speed in the horizontal axis of $p^*$, i.e., we set $\kappa_{s_1} = 1$ while the rest of coordinates are zero in (\ref{eq: ormo}). Note that the agents in red color (4 and 8) are on the horizontal axis $y=0$ in $p^*$, consequently they stop. The rest of the agents follow a constant velocity. In particular, the agents 1,2 and 3 follow the same velocity with the same speed since they are at the same $y = 1$ coordinate in $p^*$. The agents 5,6 and 7 have a (parallel) velocity with the same speed as the other three agents but with opposite sign since they are at the $y = -1$ coordinate in $p^*$. The collective behavior is the stretching of one of the affine transformations of $p^*$.}
\label{fig: aff4}
\end{figure}

	\section*{Acknowledgments}
	The work of H.G. de Marina is supported by the grant Atraccion de Talento 2019-T2/TIC-13503 from the Government of Madrid, and it has been partially supported by the Spanish Ministry of Science and Innovation under research Grant RTI2018-098962-B-C21.
	\bibliographystyle{IEEEtran}
	\bibliography{Bibs}

\begin{thebibliography}{10}
\providecommand{\url}[1]{#1}
\csname url@samestyle\endcsname
\providecommand{\newblock}{\relax}
\providecommand{\bibinfo}[2]{#2}
\providecommand{\BIBentrySTDinterwordspacing}{\spaceskip=0pt\relax}
\providecommand{\BIBentryALTinterwordstretchfactor}{4}
\providecommand{\BIBentryALTinterwordspacing}{\spaceskip=\fontdimen2\font plus
\BIBentryALTinterwordstretchfactor\fontdimen3\font minus
  \fontdimen4\font\relax}
\providecommand{\BIBforeignlanguage}[2]{{%
\expandafter\ifx\csname l@#1\endcsname\relax
\typeout{** WARNING: IEEEtran.bst: No hyphenation pattern has been}%
\typeout{** loaded for the language `#1'. Using the pattern for}%
\typeout{** the default language instead.}%
\else
\language=\csname l@#1\endcsname
\fi
#2}}
\providecommand{\BIBdecl}{\relax}
\BIBdecl

\bibitem{yang2018grand}
G.-Z. Yang, J.~Bellingham, P.~E. Dupont, P.~Fischer, L.~Floridi, R.~Full,
  N.~Jacobstein, V.~Kumar, M.~McNutt, R.~Merrifield \emph{et~al.}, ``The grand
  challenges of science robotics,'' \emph{Science Robotics}, vol.~3, no.~14, p.
  eaar7650, 2018.

\bibitem{Oh2015}
K.-K. Oh, M.-C. Park, and H.-S. Ahn, ``A survey of multi-agent formation
  control,'' \emph{Automatica}, vol.~53, pp. 424--440, 2015.

\bibitem{brinon2019multirobot}
L.~Bri{\~n}{\'o}n-Arranz, A.~Renzaglia, and L.~Schenato, ``Multirobot symmetric
  formations for gradient and hessian estimation with application to source
  seeking,'' \emph{IEEE Transactions on Robotics}, vol.~35, no.~3, pp.
  782--789, 2019.

\bibitem{anderson2008rigid}
B.~D.~O. Anderson, C.~Yu, B.~Fidan, and J.~M. Hendrickx, ``Rigid graph control
  architectures for autonomous formations,'' \emph{IEEE Control Systems
  Magazine}, vol.~28, no.~6, pp. 48--63, 2008.

\bibitem{de2019flexible}
H.~G. De~Marina and E.~Smeur, ``Flexible collaborative transportation by a team
  of rotorcraft,'' in \emph{2019 International Conference on Robotics and
  Automation (ICRA)}.\hskip 1em plus 0.5em minus 0.4em\relax IEEE, 2019, pp.
  1074--1080.

\bibitem{lin2016necessary}
Z.~Lin, L.~Wang, Z.~Chen, M.~Fu, and Z.~Han, ``Necessary and sufficient
  graphical conditions for affine formation control,'' \emph{IEEE Transactions
  on Automatic Control}, vol.~61, no.~10, pp. 2877--2891, 2016.

\bibitem{zhao2018affine}
S.~Zhao, ``Affine formation maneuver control of multiagent systems,''
  \emph{IEEE Transactions on Automatic Control}, vol.~63, no.~12, pp.
  4140--4155, 2018.

\bibitem{de2020distributed}
H.~G. de~Marina, ``Distributed formation maneuver control by manipulating the
  complex laplacian,'' \emph{Provisionally accepted in Automatica, arXiv
  preprint arXiv:2009.07625}, 2020.

\bibitem{Marina2017}
H.~G. de~Marina, B.~Jayawardhana, and M.~Cao, ``Distributed algorithm for
  controlling scale-free polygonal formations,'' \emph{IFAC World Congress
  IFAC-papersonline}, vol.~50, no.~1, pp. 1760--1765, 2017.

\bibitem{lin2014distributed}
Z.~Lin, L.~Wang, Z.~Han, and M.~Fu, ``Distributed formation control of
  multi-agent systems using complex laplacian,'' \emph{IEEE Transactions on
  Automatic Control}, vol.~59, no.~7, pp. 1765--1777, 2014.

\bibitem{gortler2010characterizing}
S.~J. Gortler, A.~D. Healy, and D.~P. Thurston, ``Characterizing generic global
  rigidity,'' \emph{American Journal of Mathematics}, vol. 132, no.~4, pp.
  897--939, 2010.

\bibitem{kelly2014class}
S.~D. Kelly and A.~Micheletti, ``A class of minimal generically universally
  rigid frameworks,'' \emph{arXiv preprint arXiv:1412.3436}, 2014.

\bibitem{friedland1975inverse}
S.~Friedland, ``On inverse multiplicative eigenvalue problems for matrices,''
  \emph{Linear Algebra and its Applications}, vol.~12, no.~2, pp. 127--137,
  1975.

\end{thebibliography}
	
\end{document}